\documentclass{amsart}
\usepackage{amsmath}
\usepackage{amssymb,amsthm}
\numberwithin{equation}{section}
\newtheorem{thm}{Theorem}[section]
\newtheorem{rem}[thm]{Remark}
\newtheorem{rem1}[thm]{Remarks}
\newtheorem{lem}[thm]{Lemma}
\newtheorem{dfn}[thm]{Definition}
\newtheorem{cor}[thm]{Corollary}
\newtheorem{pro}[thm]{Proposition}
\newtheorem{as}[thm]{Assumption}
\newtheorem{ack}{Acknowledgement}
\title{Scattering theory for multistate Schr\"odinger operators}
\author{Sohei Ashida}
\address[S. Ashida]{Department of Mathematics, Graduate School of Science, Kyoto University, Kyoto, 606-8502, Japan}
\email[]{ashida@math.kyoto-u.ac.jp}
\begin{document}
\maketitle

\bibliographystyle{plain}

\begin{abstract}
We study multistate Schr\"odinger operators related to molecular dynamics. We consider potentials which do not necessarily decay and prove absence of the singular continuous spectrum and propagation estimates which mean the scattering at speed larger than a positive constant and decay of the state with potentials higher than considered energy at infinity. We also consider the multistate Schr\"odinger operators with many-body structures. We obtain the Mourre estimate and the minimal velocity estimate for the many-body operators. The lower bound of the velocity is determined by the distance between the energy and thresholds below the energy.
\end{abstract}
\section{Introduction}\label{firstsec}

In the study of molecular dynamics the Born--Oppenheimer approximation plays a central role. In the Born--Oppenheimer approximation several electronic energy levels which are smaller than the total energy for some positions of nuclei are considered. These electronic levels are functions of the nuclear coordinates $x$ and regarded as potentials for nuclei. It is expected that if the total energy is larger than an electronic level in a certain direction at infinity, then the nuclei scatter, that is, the distance between some nuclei tends to infinity. This seems to be true even if the classical trajectory with the total energy for some electronic level is trapped in a bounded region, because of the transition between electronic levels due to their interaction. We verify this intuition in this paper under some conditions.

In this paper we consider the multistate Schr\"odinger operator
$$P=\mathrm{diag}(P_1,\dots,P_m)+R(x,D_x),$$
where $\mathrm{diag}(A_1,\dots,A_m)$ is a diagonal matrix whose elements are $A_1,\dots,A_m$ and $R(x,D_x)$ is a matrix of first order differential operators. This kind of operators are obtained by the reduction scheme for many-body problem for electrons and nuclei called Born--Oppenheimer approximation (see Klein, Martinez, Seiler and Wang \cite{KMSW}, Martinez and Messirdi \cite{MM} and Martinez and Sordoni \cite{MS}). In the Born--Oppenheimer approximation a small parameter $h$ which is the ratio of electronic and nuclear mass appears in front of the Laplacian $\Delta$ and in $R$, and the asymptotics as $h\to 0$ of the positions of the resonances of $P$  (see, e.g., \cite{Kl,Ma,Na,Ba,GM,GM2,FMW,As}) and the time evolution of the initial state corresponding to resonances (see \cite{BrMa}) were studied. The structure of the scattering matrices for short range potentials and short range off-diagonal elements at nontrapping energy was also studied by \cite{BeMa,MNS}.

Here we plan to prove absence of the singular continuous spectrum and propagation estimates for potentials which do not necessarily decay at infinity. Our method is based on the Mourre estimate for $P$. To prove the Mourre estimate at energy $\lambda$ we assume either that each diagonal operator satisfies the Mourre estimate at $\lambda$ with the same conjugate operator $A=(x\cdot p+p\cdot x)/2$, or that potentials are sum of a function homogeneous of degree zero outside a compact set and a decaying function, where $\lambda$ is not the critical value of the homogeneous potentials and the gradients of the homogeneous potentials satisfy some condition at $\lambda$. We also assume weak decay of $R(x,D_x)$. In the case of homogeneous potentials, our method to prove the Mourre estimate is similar to that in Agmon, Cruz-Sampedro and Herbst \cite{ACH} which modifies the conjugate operator from the generator of dilations. We also consider multistate Schr\"odinger operators with many-body structures in which the off-diagonal elements may not decay in all directions. We describe the structure of the thresholds and obtain the Mourre estimate with positivity determined by the distance $d(\lambda)$ from energy $\lambda$ to the nearest threshold below $\lambda$.

As an application of the Mourre estimate we prove the low velocity estimate which means that for the initial state with energy $\lambda$ at which the Mourre estimate holds scatters at speed larger than some constant. The low velocity estimate is proved using the abstract theory in Skibsted \cite{Sk} which also gives the large velocity estimate. As a consequence of the low velocity estimate we can see that states with potentials larger than $\lambda$ at infinity decay as time passes. In the many-body case the lower bound of the velocity is given by the distance $d(\lambda)$.

In Section 2, we introduce our assumptions and state our main results. In Section 3, we prove the Mourre estimates for multistate Sch\"odinger operators. In section 4, applying the Mourre estimates we obtain the propagation estimates. We determine the essential spectrum of multistate Sch\"odinger operators in the Appendix.

\section{Assumptions and results}
We consider the $m$-state Schr\"odinger operator
\begin{equation}\label{myeq2.0}
P=\mathrm{diag}(P_1,\dots,P_m)+R(x,D_x),
\end{equation}
on $\mathcal H:=\bigoplus_{j=1}^mL^2(\mathbb R^n)$ where $P_j=-\Delta+V_j(x)$ with real valued functions $V_j$ and $R(x,D_x)=(r_{j,k}(x,D_x))_{1\leq j,k\leq m}$ is a symmetric matrix of first order differential operators. As for $R(x,D_x)$ we assume the following.
\begin{as}\label{as1}
$r_{jk}=\tilde r_{jk}(x)\cdot\nabla+\hat r_{j,k}(x)$ where for some $\rho>0$ and any $l\in \mathbb N$, $\hat r_{jk}\in C^{\infty}(\mathbb R^n,\mathbb C)$ satisfies $\lvert(x\cdot\nabla)^l\hat r_{jk}(x)\rvert=o(1)$ as $\lvert x\rvert\to \infty$ and $\tilde r_{jk}\in C^{\infty}(\mathbb R^n,\mathbb C^n)$ satisfies $\lvert(x\cdot\nabla)^l(\tilde r_{jk}(x))_i\rvert=o(1)$ as $\lvert x\rvert\to \infty$ where $(\tilde r_{jk}(x))_i$ is the $i$-th component of $\tilde r_{jk}(x)$.
\end{as}
As for $P_j$, fixing $\lambda\in \mathbb R$ we suppose one of the following assumptions.
\begin{as}\label{as2}
Let $\lambda$ be a real number. For $j=1,\dots,m$, $P_j$ and $\lambda$ satisfy
\begin{itemize}
\item[(1)]For any $l\in \mathbb N$, we have $(x\cdot\nabla)^lV_j(x)\in L^{\infty}(\mathbb R^n)$.
\item[(2)]Setting $A=(x\cdot p+p\cdot x)/2$ where $p=-i\nabla$, there exist $\gamma_j,\delta_j>0$ and a compact operator $K_j$ such that denoting by $E_{P_j}(I)$ the spectral projection for $P_j$ onto the interval $I_j=(\lambda-\delta_j,\lambda+\delta_j)$ we have
\begin{equation}\label{myeq2.1}
E_{P_j}(I_j)i[P_j,A]E_{P_j}(I_j)\geq \gamma_j E_{P_j}(I_j)+K_j.
\end{equation}
\end{itemize}
\end{as}

\begin{as}\label{as3}
Let $\lambda$ be a real number. For $j=1,\dots,m$, we have $V_j(x)=\tilde V_j(x)+W_j(x)$ where $\tilde V_j$ and $W_j$ are real valued and satisfy the following:
\begin{itemize}
\item[(1)]$\tilde V_j\in C^{\infty}(\mathbb R^n)$ is homogeneous of degree zero for $\lvert x\rvert\geq 1/2$ and $\lambda$ is not the critical values of the functions $S^{n-1}\ni\omega\mapsto\tilde V_j(\omega),\ j=1,\dots,m$.
\item[(2)]Let $\omega_k\in S^{n-1},\ k=1,\dots,N$ be the all directions such that $\tilde V_j(\omega_k)=\lambda$ for some $j$ and set $J_{\omega_k}:=\{j\in \mathbb N\ \vert\ 1\leq j\leq m,\ \tilde V_j(\omega_k)=\lambda\}$. Then
\begin{equation}\label{myeq2.1.1}
\nabla \tilde V_j(\omega_k)\cdot(\sum_{l\in J_{\omega_k}}\nabla \tilde V_l(\omega_k))>0,
\end{equation}
for any $j\in J_{\omega_k}$.
\item[(3)]$W_j\in C^{\infty}(\mathbb R^n)$ satisfies $\lvert \partial^{\alpha}W_j(x)\rvert =o(\lvert x\rvert^{-\lvert \alpha\rvert})$ for any $\alpha\in \mathbb N^n$. 
\end{itemize}
\end{as}

\begin{rem}
If $\# J_{\omega_k}=1$ then \eqref{myeq2.1.1} is equivalent to $\nabla \tilde V_j(\omega_k)\neq0$.
\end{rem}

Our main result is the following theorem.
\begin{thm}\label{main}
Suppose Assumption \ref{as1} and either Assumption \ref{as2} or \ref{as3}, then the following holds.\\
\begin{itemize}
\item[(i)]There exists $\delta_0>0$ such that in $I_0=(\lambda-\delta_0,\lambda+\delta_0)$ there are only finite number of eigenvalues of $P$ and there is no singular continuous spectrum of $P$. 
\item[(ii)]If $\lambda_0\in I_0$ is not an eigenvalue of $P$, then for $f\in C_0^{\infty}(\mathbb R)$ supported in a sufficiently small neighborhood of $\lambda_0$ and any $s'>s>0$, there exist constants $\lambda', \lambda''>0$ such that
$$\left(\phi\left(\frac{x^2}{t^2}<\lambda'\right)\mathbf 1\right)e^{-itP}f(P)(\langle x\rangle^{-s'}\mathbf 1)=O(t^{-s}),$$
$$\left(\phi\left(\frac{x^2}{t^2}>\lambda''\right)\mathbf 1\right)e^{-itP}f(P)(\langle x\rangle^{-s}\mathbf 1)=O(t^{-s}),$$
as $t\to +\infty$ where $\phi(\dotsm)$ is the indicator function for $\{x\vert\dots\}$ and $\mathbf 1$ is the $m$-dimensional identity matrix, so that $g(x)\mathbf 1=\mathrm{diag}(g(x),\dots,g(x))$ for a function $g(x)$.
\end{itemize}
\end{thm}

\begin{cor}\label{maincor}
Suppose Assumption \ref{as1}, either Assumption \ref{as2} or \ref{as3} and
$$\lvert \tilde r_{jk}(x)\rvert,\lvert \hat r_{jk}(x)\rvert=O(\lvert x\rvert^{-\rho})$$
for some $\rho>0$. Let $\lambda_0$ be a number as in Theorem \ref{main} (ii).
Then if $\liminf_{x\to\infty}V_j(x)>\lambda_0$, for any $s,s'>0$ satisfying $s<s'$ and $s\leq \rho$ and any $f\in C_0^{\infty}(\mathbb R)$ supported in a sufficiently small neighborhood of $\lambda_0$ we have
$$E_{jj}e^{-itP}f(P)(\langle x\rangle^{-s'}\mathbf 1)=O(t^{-s}),$$
as $t\to +\infty$ where $E_{jj}$ is the matrix whose element in $j$-th row and $j$-th column is $1$ and the others are $0$.
\end{cor}

\begin{rem1}
(1) The following example explains the results in a typical situation. We assume $m=2$. If $V_1(x)=-C\langle x\rangle^{-1}$ where $C>0$, then the interval $(-\infty,0)$ consists of only discrete spectrum and resolvent set. If $\lambda<0$, Assumption \ref{as2} holds for $P_1$ since letting $\delta_1$ be small enough and $f_1\in C_0^{\infty}(\mathbb R^n)$ be a function such that $\mathrm{supp} f_1\subset (-\infty,0)$ and $E_{P_1}(I_1)f_1(P_1)=E_{P_1}(I_1)$, using the almost analytic extension (see the proof of Theorem \ref{Mourre2}) we can see that $f_1(P_1)-f_1(-\Delta)$ is compact and $f_1(-\Delta)=0$. Moreover, if $V_2\equiv c<\lambda$ with a constant $c$, it is easy to confirm Assumption \ref{as2} for $P_2$. Thus choosing $\lambda_0$ as in Theorem \ref{main}, $u\in \mathrm{Ran}\, f(P)(\langle x\rangle^{-s'}\mathbf{1})$ scatters by the first estimate of Theorem \ref{main} (ii) for $f\in C_0^{\infty}(\mathbb R)$ supported in sufficiently small neighborhood of $\lambda_0$. This is obvious if $R(x,D_x)\equiv 0$, since in that case $u_1\equiv 0$ and $u_2$ scatters freely. For in this case $\lambda_0\notin\sigma_{\mathrm{pp}}(P)=\sigma_{\mathrm{pp}}(P_1)$, and therefore
$$u=^t(u_1,u_2)\in \mathrm{Ran}\, f(P)=\mathrm{Ran}\, (\mathrm{diag}(f(P_1),f(P_2)))=\mathrm{Ran}\, (\mathrm{diag}(0,f(P_2))).$$
However, when $R(x,D_x)\neq 0$, $u_1$ may not be $0$. In this case as we can see by Corollary \ref{maincor}, $u_1$ vanishes being converted into $u_2$ by the interaction between electronic levels and $u_2$ scatters as in Theorem \ref{main} (ii).\\
(2) The restriction $s\leq\rho$ can not seem to be removed, since the decay rate of the interactions between the components of $u\in \mathcal H$ by the off-diagonal terms is estimated as $O(t^{-\rho})$ by Theorem \ref{main} (ii).
\end{rem1}

Next we consider multisate Sch\"odinger operators with many-body structures. In the theory of the Born--Oppenheimer approximation, since more than two nuclei can share electrons, there seem to be interactions among more than two nuclei. Therefore, it seems natural to assume that $P$ has a generalized many-body structure in which the interactions among any number of particles can be considered. To describe the generalized many-body structure we introduce some notations (see, e.g., Derezi\'nski and G\'erard \cite[chapter 5]{DG}).

Set $X=\mathbb R^n$ and suppose that
\begin{equation}\label{myeq2.2}
\{X_b\ \vert\ b\in \mathbb B\}
\end{equation} 
is a finite family of subspaces of $X$. Let
\begin{equation}\label{myeq2.3}
\{X_a\ \vert\ a\in \mathbb A\}
\end{equation}
be the smallest family of subspaces  of $X$ satisfying the following conditions:
\begin{itemize}
\item[(1)] $X$ belongs to \eqref{myeq2.3};
\item[(2)] the family \eqref{myeq2.3} is closed with respect to intersection;
\item[(3)] the family \eqref{myeq2.2} is contained in \eqref{myeq2.3}.
\end{itemize}
We endow $\mathbb A$ with a semi-lattice structure by
$$a\leq b\ \mathrm{if}\ X_a\supset X_b.$$
We denote the minimal and maximal elements in $\mathbb A$ by $a_{\min}$ and $a_{\max}$, that is
$$X_{a_{\min}}=X,\ X_{a_{max}}=\bigcap_{a\in \mathbb A}X_a.$$
We assume that $X_{a_{\max}}=\{0\}$.
We denote the orthogonal complement of $X_a$ by $X^a$. We denote by $\Pi^a$ and $\Pi_a$ the orthogonal projections of $X$ onto $X^a$ and $X_a$ respectively. We use the same notations $\Pi^a$ and $\Pi_a$ for the corresponding orthogonal projections of the dual space of $X$. We define for all $x\in X$, $x_a=\Pi_a x$ and $x^a=\Pi^a x$. We also define $\nabla_a=\Pi_a\nabla$ and $\nabla^a=\Pi^a\nabla$. The operators $-\Delta_a$ and $-\Delta^a$ denote the Laplacian in $X_a$ and $X^a$ respectively. 

We assume that for every $j=1,\dots,m$ and $b\in \mathbb B$ we are given real functions $X^b\ni x^b\mapsto v_j^b(x^b)$ and $v_j^{a_{\min}}=0$. Let $c_1,\dots,c_m$ be real numbers. We set
\begin{align*}
&V_j(x):=\sum_{b\in \mathbb B}v_j^b(x^b)+c_j,\\
&V_j^a(x^a):=\sum_{b\leq a}v_j^b(x^b)+c_j.
\end{align*}
Denoting the complexification of $X^b$ by $(X^b)^C$, for any $j,k=1,\dots,m$ let $\tilde r_{jk}^b(x^b)$ (resp., $\hat r_{jk}^b(x^b)$) be $(X^b)^C$ valued (resp., $\mathbb C$ valued) function. We assume $\tilde r_{jk}^{a_{\min}}=0,\ \hat r_{jk}^{a_{\min}}=0$. We set 
\begin{align*}
&R(x,D_x)=(r_{jk}(x,D_x))_{1\leq j,k\leq m},\ r_{j,k}(x,D_x)=\sum_{b\in \mathbb B}(\tilde r_{jk}^b(x^b)\cdot \nabla^b+\hat r_{jk}^b(x^b)),\\
&R^a(x^a,D_{x^a})=(r^a_{jk}(x^a,D_{x^a}))_{1\leq j,k\leq m},\ r^a_{jk}(x^a,D_{x^a})=\sum_{b\leq a}(\tilde r_{jk}^b(x^b)\cdot \nabla^b+\hat r_{jk}^b(x^b)).
\end{align*}

Using these notations we define $P$ as \eqref{myeq2.0} and $P^a$ for $a\neq a_{\min}$ as follows:
$$P^a:=\mathrm{diag}(P_1^a,\dots,P_m^a)+R^a(x^a,D_{x^a}),$$
where $P_j^a:=-\Delta^a+V_j^a$. We suppose the following assumptions on $v_j^b$, $\tilde r_{jk}^b$ and $\hat r_{jk}^b$.
\begin{as}\label{as4}
We assume for any $j,k=1,\dots,m,\ b\in \mathbb B$, $v_j^b,\hat r_{jk}^b\in C^{\infty}(X^b,\mathbb C)$ and $\tilde r_{jk}^b\in C^{\infty}(X^b,(X^b)^C)$. Moreover, we assume for any $l\in \mathbb N$, $\lvert(x^b\cdot\nabla^b)^l v_j^b(x^b)\rvert=o(1)$, $\lvert(x^b\cdot\nabla^b)^l \hat r_{jk}^b(x^b)\rvert=o(1)$ and $\lvert(x^b\cdot\nabla^b)^l (\tilde r_{jk}^b)_i(x^b)\rvert=o(1)$ as $\lvert x^b\rvert\to \infty$ where $(\tilde r_{jk}^b)_i,\ i=1,\dots,\mathrm{dim}X^b$ is the $i$-th component of $\tilde r_{jk}^b$ with respect to some basis of $(X^b)^C$.
\end{as}

We define the set of thresholds for $a\neq a_{\min}$ as
$$\mathcal T^a:=\bigcup_{a_{\min}<b<a}\sigma_{\mathrm{pp}}(P^b)\cup\{c_1,\dots,c_m\},$$
where $\sigma_{\mathrm{pp}}(Q)$ is the set of pure point spectrum of $Q$ and $b<a$ means $b\leq a$ and $b\neq a$. The set $\mathcal T^{a_{\mathrm{max}}}$ is simply denoted by $\mathcal T$. Moreover, we set $\Sigma^a:=\inf(\mathcal T^a)$ and $\Sigma:=\Sigma^{a_{\mathrm{max}}}=\inf(\mathcal T)$. Then exactly as in the proof of HVZ theorem (see, e.g., \cite[Theorem 6.2.2]{DG}) we have the following proposition.

\begin{pro}
Suppose Assumption \ref{as4}. Then we have 
$$\sigma_{\mathrm{ess}}(P)=[\Sigma,\infty),$$
where $\sigma_{\mathrm{ess}}(Q)$ is the essential spectrum of $Q$.
\end{pro}

\begin{rem}
In the construction of the Weyl sequence as in the proof of usual HVZ theorem, we regard $\mathcal H=\bigoplus_{j=1}^mL^2(\mathbb R^n)$ as $L^2(X_a)\otimes \left(\bigoplus_{j=1}^mL^2(X^a)\right)$.\\
\end{rem}

For $\lambda\geq \Sigma$ we define
$$d(\lambda):=\inf\{\lambda-\tau\ \vert\ \tau\leq \lambda,\ \tau\in \mathcal T\}.$$

We have the Mourre estimate and propagation estimates for $P$.
\begin{thm}\label{main2}
Suppose Assumption \ref{as4}. Then we have
\begin{itemize}
\item[(i)] $P$ does not have singular continuous spectrum, $\mathcal T\cup \sigma_{\mathrm{pp}}(P)$ is a closed countable set and $\sigma_{\mathrm{pp}}(P)$ can accumulate only at $\mathcal T$.
\item[(ii)] For any $\lambda\geq \Sigma$ and $\epsilon>0$ there exist $\delta>0$ and a compact operator $K$ such that
$$E_P(I)i[P,\mathcal A]E_P(I)\geq 2(d(\lambda)-\epsilon)E_P(I)+K,$$
where $\mathcal A=\mathrm{diag}(A,\dots,A)$, $A:=(x\cdot p+p\cdot x)/2$ and $I:=(\lambda-\delta,\lambda+\delta)$.
\item[(iii)] Let $\lambda>\Sigma$ such that $\lambda\notin\mathcal T\cup\sigma_{\mathrm{pp}}(P)$ and $\epsilon>0$ be given. Then for $f\in C_0^{\infty}(\mathbb R)$ supported in a sufficiently small neighborhood of $\lambda$ and any $s'>s>0$ we have
$$\left(\phi\left(\frac{x^2}{4t^2}<d(\lambda)-\epsilon\right)\mathbf 1\right)e^{-itP}f(P)(\langle x\rangle^{-s'}\mathbf 1)=O(t^{-s}),$$
as $t\to +\infty$ and there exists a constant $\lambda'>0$ such that
$$\left(\phi\left(\frac{x^2}{t^2}>\lambda'\right)\mathbf 1\right)e^{-itP}f(P)(\langle x\rangle^{-s}\mathbf 1)=O(t^{-s}),$$
as $t\to +\infty$.
\end{itemize}
\end{thm}

\begin{rem1}
(1) In the Mourre estimate singularities of potentials could be allowed (see \cite{DG}). In the propagation estimate (iii), proving the Mourre estimate for another conjugate operator using Graf's vector field, singularities of potentials would be allowed. Here we use the usual conjugate operator $A$ for the simple proof of the Mourre estimate.\\
(2) In general, the propagation estimate as in Corollary \ref{maincor} can not seem to hold for many-body operators because the interactions between the components of $u\in \mathcal H$ may not decay as time passes.
\end{rem1}

\section{The Mourre estimate}
In this section we prove the Mourre estimates for $P$ under the Assumptions  \ref{as1} and \ref{as2} or \ref{as3}. First, we consider the case of Assumption \ref{as3}. We can find functions $\chi_j(x)\in C^{\infty}(\mathbb R^n)$ satisfying the following conditions. $\chi_j(x)$ are homogeneous of degree zero for $\lvert x\rvert\geq 1/2$, the map $\omega\mapsto \chi_j(\omega)$ satisfies $\chi_j(\omega)=1$ in a sufficiently close neighborhood of $\omega_k$ such that $j\in J_{\omega_k}$, that is, $V_j(\omega_k)=\lambda$ and $\chi_j(\omega)=0$ in neighborhoods of the other $\omega_k$.

For $\beta>0$ we set $a(x):=(1-2\beta\sum_{j=1}^m\tilde V_j(x)\chi_j(x))\lvert x\rvert^2/4$ and
\begin{equation}\label{myeq3.0.0.0}
A_V:=i[-\Delta,a]=\nabla a\cdot p+p\cdot\nabla a,
\end{equation}
where $p=-i\nabla$. Then it is easy to see that $A_V$ is essentially selfadjoint on $C_0^{\infty}(\mathbb R^n)$. We have the Mourre estimate for $\tilde P_j:=-\Delta+\tilde V_j$ and $A_V$ at $\lambda$ in Assumption \ref{as3}.
\begin{lem}\label{Mourre1}
Suppose the Assumption \ref{as3} and let $A_V$ be as in \eqref{myeq3.0.0.0}. Then for $j=1,\dots,m$ and $\beta>0$ sufficiently small, there are constants $\tilde \delta_j,\tilde \gamma_j>$ and compact operators $K_j$ such that
$$E_{\tilde P_j}(I_j)i[\tilde P_j,A_V]E_{\tilde P_j}(I_j)\geq \tilde\gamma_jE_{\tilde P_j}(I_j)+K_j,$$
where $I_j=(\lambda-\tilde \delta_j,\lambda+\tilde \delta_j)$.
\end{lem}
\begin{proof}
As in \cite[Appendix C]{ACH} we calculate
$$i[\tilde P_j,A_V]=-4\sum_{k,l=1}^n\partial_ka_{kl}^{(2)}(x)\partial_l+\beta\lvert x\rvert^2\sum_{k=1}^m\nabla \tilde V_j(x)\cdot\nabla(\tilde V_k(x)\chi_k(x))+G(x),$$
where
$$a_{kl}^{(2)}=\frac{1}{2}\delta_{kl}-\frac{\beta}{2}\partial_k\partial_l\left(\lvert x\rvert^2\sum_{i=1}^m\tilde V_i(x)\chi_i(x)\right),$$
$$G(x)=-(\Delta^2a(x))-\left(1-2\beta\sum_{k=1}^m\tilde V_k(x)\chi_k(x)\right)x\cdot\nabla\tilde V_j(x).$$
Since $x\cdot\nabla\tilde V_j(x)=0$ for $\lvert x\rvert \geq 1$, we have $G(x)\to 0$ as $\lvert x\rvert \to\infty$. For sufficiently small $\beta$ we also have $-4\sum_{k,l=1}^n\partial_ka_{kl}^{(2)}(x)\partial_l\geq -\Delta$. Thus it remains to show
\begin{equation}\label{myeq3.00}
E_{\tilde P_j}(I_j)\left(-\Delta+\beta\lvert x\rvert^2\sum_{k=1}^m\nabla \tilde V_j(x)\cdot\nabla(\tilde V_k(x)\chi_k(x))\right)E_{\tilde P_j}(I_j)\geq\tilde \gamma_jE_{\tilde P_j}(I_j)+K,
\end{equation}
for some $\tilde \gamma_j>0$ and some compact operator $K$.

By the definition of $\chi_j$ and Assumption \ref{as3}, there exist $\epsilon>0$ and $\hat \gamma_j>0$ such that the following holds. We have $\{\omega\in S^{n-1}\ \vert\ \lambda-2\epsilon<\tilde V_j(\omega)<\lambda+2\epsilon\}=\bigcup_{j\in J_{\omega_k}}\Omega_k$ where $\Omega_k\in S^{n-1}$ are disjoint and satisfy the following:
\begin{itemize}
\item[(1)] $\omega_k\in \Omega_k$
\item[(2)] If $\omega\in \Omega_k$, we have $\chi_l(\omega)=1$ for $l\in J_{\omega_k}$ and  $\chi_l(\omega)=0$ for $l\notin J_{\omega_k}$
\item[(3)] For $\omega\in \Omega_k$ we have
\begin{equation}\label{myeq3.01}
\nabla\tilde V_j(\omega)\cdot\sum_{p\in J_{\omega_k}}\nabla\tilde V_p(\omega)>\hat\gamma_j.
\end{equation}
\end{itemize}

As in \cite[Appendix C]{ACH} the sphere $S^{n-1}$ is the union of open sets:
\begin{align*}
&\mathcal O_1=\{\omega\in S^{n-1}\ \vert \ \lambda-2\epsilon<\tilde V_j(\omega)<\lambda+2\epsilon\},\\
&\mathcal O_2=\{\omega\in S^{n-1}\ \vert \ \tilde V_j(\omega)<\lambda-\epsilon\},\\
&\mathcal O_3=\{\omega\in S^{n-1}\ \vert \ \lambda+\epsilon<\tilde V_j(\omega)\}.
\end{align*}
There is a partition of unity of $\mathbb R^n$: $\varphi_1+\varphi_2+\varphi_3+\eta^2=1$, where $\varphi_k$ is homogeneous of degree $0$ for $\lvert x\rvert>3/4$, $\mathrm{supp}\, \eta\subset\{x\ \vert \ \lvert x\rvert<1\}$ and $\mathrm{supp}\, \varphi_k\subset\{x\ \vert \ \lvert x\rvert>1/2,\ x/\lvert x\rvert\in \mathcal O_k\}$.

Let $f\in C_0^{\infty}(\mathbb R)$ be a function such that $f=1$ near a small neighborhood of $\lambda$. Since $E_{\tilde P_j}(I_j)f(\tilde P_j)=E_{\tilde P_j}(I_j)$ for $\tilde \delta_j$ small enough, we only need to prove \eqref{myeq3.00} replacing $E_{\tilde P_j}(I_j)$ by $f(\tilde P_j)$. Moreover we can choose $f$ such that $\mathrm{supp}\, f\subset [\lambda-\epsilon/2,\lambda+\epsilon/2]$.

In the following we denote compact operators by $\tilde K_l, l=1,2,\dots$. We can write
\begin{align*}
&f(\tilde P_j)\left(-\Delta+\beta\lvert x\rvert^2\sum_{k=1}^m\nabla \tilde V_j(x)\cdot\nabla(\tilde V_k(x)\chi_k(x))\right)f(\tilde P_j)\\
&\quad=\sum_{k=1}^3f(\tilde P_j)\varphi_k(x)\left(-\Delta+\beta\lvert x\rvert^2\sum_{k=1}^m\nabla \tilde V_j(x)\cdot\nabla(\tilde V_k(x)\chi_k(x))\right)\varphi_k(x)f(\tilde P_j)+\tilde K_1.
\end{align*}
In the support of $\varphi_1$ using \eqref{myeq3.01} we have,
\begin{equation}\label{myeq3.02}
\begin{split}
&f(\tilde P_j)\varphi_1(x)\left(-\Delta+\beta\lvert x\rvert^2\sum_{k=1}^m\nabla \tilde V_j(x)\cdot\nabla(\tilde V_k(x)\chi_k(x))\right)\varphi_1(x)f(\tilde P_j)\\
&\quad\geq\beta\hat\gamma_jf(\tilde P_j)\varphi_1(x)^2f(\tilde P_j).
\end{split}
\end{equation}
As for $\phi_2$ and $\phi_3$, choosing sufficiently small $\beta$ we can obtain in the same way as in \cite[Appendix C]{ACH},
\begin{equation}\label{myeq3.03}
\begin{split}
&f(\tilde P_j)\varphi_k(x)\left(-\Delta+\beta\lvert x\rvert^2\sum_{k=1}^m\nabla \tilde V_j(x)\cdot\nabla(\tilde V_k(x)\chi_k(x))\right)\varphi_k(x)f(\tilde P_j)\\
&\quad\geq(\epsilon/3)f(\tilde P_j)\varphi_k(x)^2f(\tilde P_j)+\tilde K_k,\ k=2,3.
\end{split}
\end{equation}
Choosing $\tilde \gamma_j=\min\{\beta\hat\gamma_j,\epsilon/3\}$ and adding \eqref{myeq3.02} and \eqref{myeq3.03} we obtain \eqref{myeq3.00} with $E_{\tilde P_j}(I_j)$ replaced by $f(\tilde P_j)$.
\end{proof}

Using Lemma \ref{Mourre1} we can obtain the Mourre estimate for $P$.
\begin{thm}\label{Mourre2}
Suppose the Assumptions \ref{as1} and \ref{as3} and let $A_V$ be as in \eqref{myeq3.0.0.0}. Then there are constants $\delta,\gamma_0>0$ and a compact operator $K$ such that
\begin{equation}\label{myeq3.0}
E_{P}(I)i[P,\mathcal A_V]E_{P}(I)\geq \gamma_0E_{P}(I)+K,
\end{equation}
where $I=(\lambda-\delta,\lambda+\delta)$ and $\mathcal A_V:=\mathrm{diag}(A_V,\dots,A_V)$.
\end{thm}
\begin{proof}
Let $f\in C_0^{\infty}(\mathbb R)$ be a function such that $f=1$ near a small neighborhood of $\lambda$. As in the proof of Lemma \ref{Mourre1}, we only need to prove \eqref{myeq3.0} replacing $E_P(I)$ by $f(P)$.  In the following we denote compact operators by $K_j,\ j=1,2,\dots$. Note that a matrix of operators is compact if and only if all its elements are compact. Since $D(P)=\bigoplus_{j=1}^mH^2(\mathbb R^n)$, by the Assumptions \ref{as1} and \ref{as3} on the decay of $R(x,D_x)$ and $W_j$, it is easy to see that $f(P)[\mathrm{diag}(W_1,\dots,W_m)+R(x,D_x),\mathcal A_V]f(P)$ is a compact operator. Thus setting $\tilde P:=\mathrm{diag}(-\Delta+\tilde V_1,\dots,-\Delta+\tilde V_m)=\mathrm{diag}(\tilde P_1,\dots,\tilde P_m)$ we have
\begin{equation}\label{myeq3.0.1}
f(P)i[P,\mathcal A_V]f(P)\geq f(P)i[\tilde P,\mathcal A_V]f(P)+K_1,
\end{equation}
where $K$ is compact.

Let $F$ be an almost analytic extension of $f$ (see, e.g., \cite{DG}). Then we have
\begin{equation}\label{myeq3.1}
\begin{split}
f(P)-f(\tilde P)&=\frac{1}{2\pi i}\int \bar\partial_{z}F(z)((z-P)^{-1}-(z-\tilde P)^{-1})dz\wedge d\bar z\\
&=\frac{1}{2\pi i}\int \bar\partial_{z}F(z)(z-P)^{-1}(\mathrm{diag}(W_1,\dots,W_m)+R(x,D_x))\\
&\quad \cdot(z-\tilde P)^{-1}dz\wedge d\bar z.
\end{split}
\end{equation}
By Assumptions \ref{as1} and \ref{as3} we can see easily that the right-hand side of \eqref{myeq3.1} is compact.  Thus we have
\begin{equation}\label{myeq3.2}
f(P)i[\tilde P,\mathcal A_V]f(P)=f(\tilde P)[\tilde P,\mathcal A_V]f(\tilde P)+K_2.
\end{equation}
By the uniqueness of the functional calculus, we can see that
$$f(\tilde P)=\mathrm{diag}(f(\tilde P_1),\dots,f(\tilde P_m)).$$
Therefore, by Lemma \ref{Mourre1} we can see that for $f$ supported in sufficiently small neighborhood of $\lambda$ with $\gamma_0=\min\{\tilde\gamma_1,\dots,\tilde \gamma_m\}$,
\begin{equation}\label{myeq3.3}
f(\tilde P)[\tilde P,\mathcal A_V]f(\tilde P)\geq \gamma_0 f(\tilde P)^2+K_3.
\end{equation}
Using \ref{myeq3.1} again, we obtain
\begin{equation}\label{myeq3.4}
\gamma_0 f(\tilde P)^2=\gamma_0 f(P)^2+K_4.
\end{equation}
Combining \eqref{myeq3.0.1}, \eqref{myeq3.2}, \eqref{myeq3.3} and \eqref{myeq3.4} we obtain the theorem.
\end{proof}

In the case of Assumption \ref{as2} we can prove the Mourre estimate in the same way as Theorem \ref{Mourre2} using Assumption \ref{as2} (2) instead of Lemma \ref{Mourre1}.
\begin{thm}\label{Mourre2'}
Suppose the Assumptions \ref{as1} and \ref{as2} and set $A=(x\cdot p+p\cdot x)/2,$ where $p:=-i\nabla$. Then there is a constant $\delta>0$ and a compact operator $K$ such that
\begin{equation*}
E_{P}(I)i[P,\mathcal A]E_{P}(I)\geq \gamma_0E_{P}(I)+K,
\end{equation*}
where $I=(\lambda-\delta,\lambda+\delta)$, $\gamma_0:=\min\{\gamma_1,\dots,\gamma_m\}$ and $\mathcal A:=\mathrm{diag}(A,\dots,A)$.
\end{thm}

By Theorem \ref{Mourre2} (resp., Theorem \ref{Mourre2'}), Assumptions \ref{as1} and \ref{as2} (resp., \ref{as3}), we can prove that there exist the limit of the resolvent $R(z)$ of $P$ as  operators between certain Banach spaces as $z$ tends to continuous spectrum applying the method of Mourre \cite{Mo}. As a result, using Stone's formula we obtain Theorem \ref{main} (i).

Finally, we consider the many-body case.
\begin{proof}[Proof of Theorem \ref{main2} (i) and (ii)]
The proof of Theorem \ref{main2} (i) and (ii) is very similar to that of usual many-body Hamiltonian (see R\cite[section 6.4]{DG}). The  differences are the following points. Set $P_{a_{\min}}:=\mathrm{diag}(-\Delta+c_1,\dots,-\Delta+c_m)$ and for $\ \lambda\geq\min\{c_1,\dots,c_m\}$
$$d^{a_{\min}}(\lambda):=\min\{\lambda-c_j\ \vert\ c_j\leq \lambda\}.$$
Then given $\epsilon>0$ and $\lambda\in \mathbb R$, there exists $\delta>0$ such that
\begin{equation}\label{myeq3.5}
E_{P_{a_{\min}}}(I)i[P_{a_{\min}},\mathcal A]E_{P_{a_{\min}}}(I)\geq (d^{a_{\min}}(\lambda)-\epsilon)E_{P_{a_{\min}}}(I),
\end{equation}
where $I=(\lambda-\delta,\lambda+\delta)$. The estimate \eqref{myeq3.5} follows from that
$$E_{P_{a_{\min}}}(I)=\mathrm{diag}(E_{-\Delta+c_1}(I),\dots,E_{-\Delta+c_m}(I)),$$
that  for $\lambda\geq c_j$ we have
$$E_{-\Delta+c_j}(I)i[-\Delta+c_j,A]E_{-\Delta+c_j}(I)\geq2(\lambda-c_j-\delta)E_{-\Delta+c_j}(I),$$
and that if $\lambda< c_j$, for  $\delta<c_j-\lambda$ and any $\alpha>0$ we have
$$E_{-\Delta+c_j}(I)i[-\Delta+c_j,A]E_{-\Delta+c_j}(I)\geq\alpha E_{-\Delta+c_j}(I),$$
since $E_{-\Delta+c_j}(I)=0$.

Another difference is that we use the decay not only of $(I_j)_a(x):=V_j(x)-V_j^a(x^a)$ but also of $R_a(x,D_x):=R(x,D_x)-R^a(x^a,D_{x^a})$ when $\lvert x^b\rvert\to \infty$ for all $b\nleq a$. We omit the details of the proof.
\end{proof}

\section{Propagation estimates}
In this section we prove Theorem \ref{main} (2), Corollary \ref{maincor} and Theorem \ref{main2}. We apply the abstract theory of propagation estimates in Skibsted \cite{Sk}. We need the following class of functions with a parameter $\tau$ (see \cite[Definition 2.1]{Sk}).
\begin{dfn}\label{dfn1}
Given $\beta,\alpha\geq 0$ and $\epsilon>0$ let $\mathcal F_{\beta,\alpha,\epsilon}$ denote the set of functions $g,\ g(x,\tau)=g_{\beta,\alpha,\epsilon}(x,\tau):=-\tau^{-\beta}(-x)^{\alpha}\chi\left(\frac{x}{\tau}\right),$ defined for $(x,\tau)\in \mathbb R\times \mathbb R^+$ and for $\chi\in C^{\infty}(\mathbb R)$ with the following properties:
\begin{align*}
\chi(x)=1\ \mathrm{for} \ x<-2\epsilon,\ \chi(x)=0\ \mathrm{for}\ x>-\epsilon,\\
\frac{d}{dx}\chi(x)\leq 0\ \mathrm{and}\ \chi(x)+x\frac{d}{dx}\chi(x)=\tilde\chi^2(x),
\end{align*}
where $\tilde \chi(x)\geq 0$ and $\tilde \chi(x)\in C^{\infty}(\mathbb R)$.
\end{dfn}

In the proof of the propagation estimates we shall use the notation $\chi(G<-\epsilon):=\chi(G)$ for a selfadjoint operator $G$ with $\chi$ smooth and satisfying the first three properties enlisted above.

If we confirm Assumption 2.2 (1)-(5) and the assumption in Corollary 2.6 in \cite{Sk} replacing $H$ by $P$ for some selfadjoint operators $A(\tau),B$ and $n_0\geq2,t_0=1,\kappa_0=0,\beta_0>0,n_0-1/2>\alpha_0>0$ applying Theorem 2.4 in \cite{Sk} we obtain for $(\beta,\alpha)=(0,1),\dots,(0,\alpha_0'),(\beta_0,\alpha_0)$ any $\epsilon>0$ and $g(x,\tau)\in\mathcal F_{\beta,\alpha,\epsilon}$
\begin{equation}\label{myeq4.0}
(-g_{\beta,\alpha,\epsilon}(A(\tau),\tau))^{1/2}e^{-itP}f(P)B^{-\alpha/2}=O(1),
\end{equation}
in $\mathcal L(\mathcal H)$ as $t\to \infty$ where $\tau=t+1$ and $\alpha_0':=\max\{n\in \mathbb N\ \vert\ n<\alpha_0\}$.

Since in the following cases to confirm Assumption 2.2 (1)-(5) for $\kappa_0=0$, $t_0=1$ and any $n_0\geq2,\beta_0>0,n_0-1/2>\alpha_0>0$ is easy, we see only the assumption in Corollary 2.6, that is, that
for several operators $A(\tau)$, for any $\beta_0,\alpha_0>0$ and for $f_1\in C_0^{\infty}(\mathbb R)$ supported in sufficiently small neighborhood of $\lambda_0$ in the assumption of Theorem \ref{main} (ii), there exist some $\delta>0$ and $C>0$ such that
\begin{equation}\label{myeq4.0.1}
f_1(P)DA(\tau)f_1(P)\geq B_1(\tau)+B_2(\tau),
\end{equation}
$$B_1(\tau)=O(\tau^{-\delta}),$$
and 
$$\int_0^{\infty}dt\lvert\langle \zeta(t),B_2(\tau)\zeta(t)\rangle\rvert\leq C\lVert \phi\rVert^2,\ \forall \phi\in \mathcal H,$$
where $D A(\tau)=i[P,A(\tau)]+d_tA(\tau)$, $\zeta(t)=(g^{(1)}(A(\tau),\tau))^{1/2}e^{-itP}f(P)B^{-\alpha/2}\phi$, $g(x,\tau)\in \mathcal F_{\beta,\alpha,\epsilon}$ for $(\beta,\alpha)=(0,1),\dots,(0,\alpha_0'),(\beta_0,\alpha_0)$.

\begin{proof}[Proof of Theorem \ref{main} (ii)]
We consider the case of Assumption \ref{as3}. Let $\mathcal A_V$ and $\gamma_0$ be as in Theorem \ref{Mourre2}. Fix $0<\gamma_0'<\gamma_0$ and as in \cite[Example 1]{Sk} set for $\tau>0$, $\tilde{\mathcal A}(\tau):=\mathcal A_V-\gamma_0'\tau\mathbf{1}$. First we shall prove the estimate for $\tilde{\mathcal A}(\tau)$. Using Theorem \ref{Mourre2} and noting that $E_P(\tilde I)K\to0$ as $\tilde \delta\to0$ where $\tilde I:=[\lambda_0-\tilde \delta,\lambda_0+\tilde \delta]$, we can see that $f_1(P)D\tilde{\mathcal A}(\tau)f_1(P)\geq 0$ for $f_1\in C_0^{\infty}(\mathbb R)$ supported in a small neighborhood of $\lambda_0$. Thus  we obtain \eqref{myeq4.0.1}, and therefore \eqref{myeq4.0} for any $\beta_0,\alpha_0>0$ with $A(\tau)$ and $B$ replaced by $\tilde{\mathcal A}(\tau)$ and $\langle\mathcal A_V\rangle$ respectively where $\langle G\rangle:=(\mathbf{1}+G^2)^{1/2}$. Since for $\beta,\alpha\geq 0$ and $1\geq\theta\geq0$ we have
\begin{equation}\label{myeq4.0.2}
-g_{\beta,\alpha,\epsilon}(x,\tau)\geq-\tau^{-\beta}(\epsilon\tau)^{\alpha\theta}g_{0,\alpha(1-\theta),2\epsilon}(x,\tau),
\end{equation}
(see Corollary 2.5 in \cite{Sk}), we obtain
\begin{equation}\label{myeq4.1}
\left(\frac{-\tilde{\mathcal A}(\tau)}{\tau}\right)^{1/2}\chi\left(\frac{\tilde{\mathcal A}(\tau)}{\tau}<-\epsilon\right)e^{-itP}f(P)\langle \mathcal A_V\rangle^{-s}=O(\tau^{-s}),
\end{equation}
for any $s\geq 1/2$ and $\epsilon>0$.

We shall prove the first estimate in Theorem \ref{main} (ii). Given $\epsilon''>0$ let $g(x,\tau)\in \mathcal F_{0,1,\epsilon''}$, that is,
$$g(x,\tau)=x\chi\left(\frac{x}{\tau}\right),$$
as in Definition \ref{dfn1}.
Fix $0<\gamma_0''<\gamma_0'$ and set $\mathcal A'(\tau):=\mathrm{diag}(g(-\tau M,\tau),\dots,g(-\tau M,\tau))$, where $M=M(x,\tau):=\left(\frac{\gamma_0''}{2}-\frac{a(x)}{\tau^2}\right)^{1/2}$.

Then as in \cite[Example 2]{Sk} we have
\begin{equation}\label{myeq4.1.1}
\begin{split}
D\mathcal A'(\tau)=&\left(\frac{\partial}{\partial \tau}g\right)(-\tau M,\tau)+\frac{1}{2}((g^{(1)}(-\tau M,\tau))^{1/2}M^{-1/2}\mathbf{1})\\
&\cdot \left\{\frac{\tilde{\mathcal A}(\tau)}{\tau}+(\gamma_0'-\gamma_0'')\mathbf{1}+\frac{2i}{\tau}[R(x,D_x),a(x)\mathbf{1}]\right\}(M^{-1/2}(g^{(1)}(-\tau M,\tau))^{1/2}\mathbf{1})\\
&+O(\tau^{-1}).
\end{split}
\end{equation}

The first term is nonnegative. As for the second term let $\mathcal R:=\frac{2i}{\tau}[R(x,D_x),a(x)\mathbf{1}]$. Then the element in the $j$-th row and $k$-th column of $\mathcal R$ is $\mathcal R_{jk}=\frac{2i}{\tau}\tilde r_{jk}\cdot\nabla a(x)$.
Let $\hat \chi\in C^{\infty}(\mathbb R)$ a function supported in $(-\infty,-\epsilon''/2)$ such that $\hat\chi=1$ on $(-\infty,-\epsilon'']$. Then we have
$$g^{(1)}(-\tau M,\tau)=\hat\chi(-M)g^{(1)}(-\tau M,\tau).$$
If $\beta$ in $a(x)$ is small enough, we have
\begin{equation*}
\frac{\lvert x\rvert^2}{\tau^2}<C\frac{a(x)}{\tau^2}<C(\gamma_0''-\epsilon''^2/4)
\end{equation*}
for some $C>0$ on $\mathrm{supp}\, \hat\chi(-M)$.  Thus noting that $\lvert \nabla a(x)\rvert\leq C\lvert x\rvert$ for some $C>0$, we can see that there exists $C_0>0$ such that
\begin{equation}\label{myeq4.1.2}
\tau^{-1}\lvert\nabla a\rvert\hat\chi(-M)<C_0,
\end{equation}
uniformly with respect to $\tau$ and $x$.

Let $f_2\in C_0^{\infty}(\mathbb R)$ be a function such that $f_2f_1=f_1$. Then since we have
$$(1-f_2(P))(M^{-1/2}(g^{(1)}(-\tau M,\tau))^{1/2}\mathbf{1})f_1(P)=O(\tau^{-1}),$$
we obtain
\begin{equation}\label{myeq4.1.3}
\begin{split}
&f_1(P)((g^{(1)})^{1/2}M^{-1/2}\mathbf{1})\mathcal R(M^{-1/2}(g^{(1)})^{1/2}\mathbf{1})f_1(P)\\
&\quad =f_1(P)((g^{(1)})^{1/2}M^{-1/2}\mathbf{1})f_2(P)(\hat\chi(-M)\mathbf 1)\mathcal R(\hat\chi(-M)\mathbf 1))f_2(P)\\
&\qquad \cdot(M^{-1/2}(g^{(1)})^{1/2}\mathbf{1})f_1(P)+O(\tau^{-1}).
\end{split}
\end{equation}
Letting $f_3\in C_0^{\infty}(\mathbb R)$ be a function such that $f_3f_2=f_2$, by \eqref{myeq4.1.2} and Assumption \ref{as1} we can see that
\begin{equation}\label{myeq4.1.4}
f_3(P)(\hat\chi(-M)\mathbf 1)\mathcal R(\hat\chi(-M)\mathbf 1))f_3(P)\geq -C_1 K,
\end{equation}
for some $C_1>0$ and a compact operator $K$.

If the support of $f_2$ is small enough, we have $-C_1f_2(P)Kf_2(P)>-\tilde\epsilon\mathbf 1$ with $\tilde\epsilon>0$ satisfying $\gamma_0'-\gamma_0''>\tilde \epsilon$. Therefore, combining \eqref{myeq4.1.1}, \eqref{myeq4.1.3} and \eqref{myeq4.1.4} we obtain
\begin{align*}
f_1(P)D\mathcal A'(\tau)f_1(P)\geq& f_1(P)((g^{(1)})^{1/2}M^{-1/2}\mathbf{1})\left\{\frac{\tilde{\mathcal A}(\tau)}{\tau}+(\gamma_0'-\gamma_0''-\tilde\epsilon)\mathbf{1}\right\}\\
&\cdot(g^{(1)}(-\tau M,\tau))^{1/2}\mathbf{1})f_1(P)+O(\tau^{-1}).
\end{align*}
Thus as in \cite[Example 2]{Sk} observing
$$\frac{\tilde{\mathcal A}(\tau)}{\tau}+(\gamma_0'-\gamma_0''-\tilde\epsilon)\mathbf1\geq\frac{\tilde{\mathcal A}(\tau)}{\tau}\chi^2\left(\frac{\tilde{\mathcal A}(\tau)}{\tau}<-\epsilon'\right),$$
with $\epsilon'=(\gamma_0'-\gamma_0''-\tilde\epsilon)/2$ and using \eqref{myeq4.1}, we obtain \eqref{myeq4.0.1} and therefore \eqref{myeq4.0} for any $\beta_0,\alpha_0>0$ with $A(\tau)$ and $B$ replaced by $\mathcal A'(\tau)$ and $\mathcal \langle \mathcal A_V\rangle^{1+\kappa}$ respectively where $\kappa>0$ is arbitrary.

Therefore, noting that $g_{\beta,\alpha,\epsilon}(\mathcal A'(\tau),\tau)=g_{\beta,\alpha,\epsilon}(-\tau M\mathbf{1},\tau)$ for any $\epsilon>2\epsilon''$ we obtain by \eqref{myeq4.0.2}
$$\left(\chi\left(\frac{a(x)}{t^2}-\frac{\gamma_0''}{2}<-\epsilon\right)\mathbf{1}\right)e^{-itP}f(P)\langle \mathcal A_V\rangle^{-s(1+\kappa)}=O(t^{-s}),$$
for any $\epsilon,\kappa,s>0$. Since $\mathcal A_V^n(P-i)^{-n}(\langle x\rangle^{-n}\mathbf{1})$ is bounded for any $n\in\mathbb N$, and we have $\frac{a(x)}{\tau^2}<C\frac{\lvert x\rvert^2}{\tau^2}$, choosing $\epsilon$ small enough we obtain the first estimate of Theorem \ref{main} (ii) with $\lambda'=(\gamma_0''/2-2\epsilon)/C$.

Next we consider the second estimate. For $v>0$ we set $\hat{\mathcal A}(\tau):=(v\tau-\langle x\rangle)\mathbf{1}$. Then as in \cite[Example 3]{Sk} for $f_1\in C_0^{\infty}(\mathbb R)$ and $v$ large enough we have
$$f_1(P)D\hat{\mathcal A}(\tau)f_1(P)\geq 0,$$
and obtain \eqref{myeq4.0} for any $\beta_0,\alpha_0>0$ with $A(\tau)$ and $B$ replaced by $\hat{\mathcal A}(\tau)$ and $\langle x\rangle$ respectively. Therefore, by \eqref{myeq4.0.2} the second estimate of Theorem \ref{main} (ii) holds. In the same way as above we can obtain the estimates for Assumption \ref{as2}.
\end{proof}

\begin{proof}[Proof of Corollary \ref{maincor}]
By the assumption there exists a function $W(x)\in C_0^{\infty}(\mathbb R^n)$ such that $\inf_{x\in \mathbb R^n}(V_j(x)+W(x))>\lambda_0$. Set $\hat P:=\mathrm{diag}(P_1,\dots,P_j+W,\dots,P_m)$. Let $f_1\in C_0^{\infty}(\mathbb R)$ such that $f_1f=f$. Then we can write
\begin{equation}\label{myeq4.2}
\begin{split}
e^{-itP}f(P)(\langle x\rangle^{-s'}\mathbf{1})=&f_1(\hat P)e^{-itP}f(P)(\langle x\rangle^{-s'}\mathbf{1})\\
&+(f_1(P)-f_1(\hat P))e^{-itP}f(P)(\langle x\rangle^{-s'}\mathbf{1}).
\end{split}
\end{equation}
As for the first term in the right-hand side of \eqref{myeq4.2} we have
$$f_1(\hat P)=\mathrm{diag}(f_1(P_1),\dots,f_1(P_j+W),\dots,f_1(P_m)).$$
Since $P_j+W>\lambda_0$, we have $f_1(P_j+W)=0$ if $f$ and $f_1$ are supported in a sufficiently small neighborhood of $\lambda_0$. Thus we have $E_{jj}f_1(\hat P)=0$.

As for the second term using almost analytic extension $F_1$ of $f_1$ we have
\begin{equation*}
\begin{split}
f_1(P)-f_1(\hat P)&=\frac{1}{2\pi i}\int \bar\partial_{z}F_1(z)(z-P)^{-1}(R(x,D_x)-WE_{jj})\\
&\quad \cdot(z-\hat P)^{-1}dz\wedge d\bar z.
\end{split}
\end{equation*}
By the assumption and the definition of $W$ we have $(f_1(P)-f_1(\hat P))(\langle x\rangle^{\rho}\mathbf{1})\in \mathcal L(\mathcal H)$. Thus we only need to prove
\begin{equation}\label{myeq4.3}
(\langle x\rangle^{-\rho}\mathbf{1})e^{-itP}f(P)(\langle x\rangle^{-s'})=O(t^{-s}),
\end{equation}
as $t\to +\infty$ for $s,s'>0$ such that $s<s'$ and $s\leq \rho$.

We can write
\begin{align*}
(&\langle x\rangle^{-\rho}\mathbf{1})e^{-itP}f(P)(\langle x\rangle^{-s'})\\
&=(\langle x\rangle^{-\rho}\mathbf{1})\left(\chi\left(\frac{x^2}{t^2}< \lambda'\right)\mathbf{1}\right)e^{-itP}f(P)(\langle x\rangle^{-s'})\\
&+(\langle x\rangle^{-\rho}\mathbf{1})\left(\chi\left(\frac{x^2}{t^2}\geq \lambda'\right)\mathbf{1}\right)e^{-itP}f(P)(\langle x\rangle^{-s'}).
\end{align*}
Choosing $f$ supported in a sufficiently small neighborhood of $\lambda_0$ and $\lambda'$ as in Theorem \ref{main} (ii), the first term is estimated as $O(t^{-s})$ for $0<s<s'$. By the estimate
$$\langle x\rangle^{-\rho}\chi\left(\frac{x^2}{t^2}>\lambda'\right)=O(t^{-\rho}),$$
the second term is estimated as $O(t^{-\rho})$, and we obtain \eqref{myeq4.3}.
\end{proof}

\begin{proof}[Proof of Theorem \ref{main2} (iii)]
The proof of Theorem \ref{main2} (iii) is similar to that of Theorem \ref{main} (ii). The estimate \eqref{myeq4.1} is proved for $\tilde {\mathcal A}(\tau)=\mathcal A-\gamma_0'\tau\mathbf{1}$ with any $0<\gamma_0'<2d(\lambda)$ and $\mathcal A=\mathrm{diag}(A,\dots,A)$ where $A:=(x\cdot p+p\cdot x)/2$.

In the definition of $M$ we replace $a(x)$ by $\frac{S^2}{2}G\left(\frac{x}{S}\right)$ where $S>0$ and $G(x)$ is the convex function whose gradient is the Graf's vector field (see, e.g., \cite[Lemma 5.2.7]{DG}). We collect the properties of $G$ we use.
\begin{lem}\label{Graf}
$G(x)$ is a smooth convex function and satisfy the following:
\begin{itemize}
\item[(1)] $\max\{x^2,C_1\}\leq 2G(x)\leq x^2+C_2$ for some $C_1,C_2>0$.
\item[(2)] \begin{equation}\label{myeq4.4}
\partial^{\alpha}(2G(x)-x^2)\in L^{\infty}(\mathbb R^n),\ \forall\alpha\in \mathbb N^n,
\end{equation}
\item[(3)]  there exists a constant $\delta>0$ such that for any $a\in \mathbb A$, if $\lvert x^a\rvert<\delta$, $G(x)$ depends only on $x_a$.
\end{itemize}
\end{lem}

The equation \eqref{myeq4.1.1} is obtained with $0<\gamma_0''<\gamma_0'$ and the factors $\tilde{\mathcal A}(\tau)$ and  $\frac{2i}{\tau}[R(x,D_x),a(x)\mathbf{1}]$ replaced by
$$\tilde{\mathcal A}'(\tau):=i\left[\mathrm{diag}(-\Delta\dots,-\Delta),\frac{S^2}{2}G\left(\frac{x}{S}\right)\mathbf{1}\right]-\gamma_0'\tau\mathbf{1}$$
and $\tilde{\mathcal R}:=\frac{2i}{\tau}\left[R(x,D_x),\frac{S^2}{2}G\left(\frac{x}{S}\right)\mathbf{1}\right]$ respectively. As for $\tilde{\mathcal A}'(\tau)$ by \eqref{myeq4.4} we get $\lVert\tau^{-1}(\tilde{\mathcal A}'(\tau)-\tilde{\mathcal A}(\tau))f_2(P)\rVert_{\mathcal L(\mathcal H)}\leq CS\tau^{-1}$ for some $C>0$ with $f_2$ as in the proof of Theorem \ref{main} (ii). Thus we can replace $\tilde{\mathcal A}'(\tau)$ by $\tilde{\mathcal A}(\tau)$.

As for $\tilde{\mathcal R}$, the element in $j$-th row and $k$-th column of $\tilde{\mathcal R}$ is
$$\tilde{\mathcal R}_{jk}=\frac{iS}{\tau}\sum_{b\in \mathbb B}\tilde r_{jk}^b(x^b)\nabla^bG\left(\frac{x}{S}\right).$$
Let $\hat\chi$ be the function introduced in the proof of Theorem \ref{main} (ii). We write
\begin{equation}\label{myeq4.5}
\frac{iS}{\tau}\nabla^bG\left(\frac{x}{S}\right)\hat\chi(-M)=\frac{i}{\tau}\left(S\nabla^bG\left(\frac{x}{S}\right)-x^b\right)\hat\chi(-M)+\frac{i}{\tau}x^b\hat\chi(-M).
\end{equation}
Then by \eqref{myeq4.4} we have
$$\left\lvert\frac{i}{\tau}\left(S\nabla^bG\left(\frac{x}{S}\right)-x^b\right)\right\rvert\leq CS\tau^{-1}$$ for some $C>0$. However,  by Lemma \ref{Graf} (1) on $\mathrm{supp}~ \hat\chi(-M)$ we have 
$$C_1\frac{S^2}{\tau^2}\leq\frac{S^2}{\tau^2}G\left(\frac{x}{S}\right)\leq C,$$
for some $C>0$. Thus the first term in the right-hand side of \eqref{myeq4.5} is uniformly bounded with respect to $S$ and $\tau$. Similarly,  by the first inequality in Lemma \ref{Graf} (1), on $\mathrm{supp}~ \hat\chi(-M)$ we have $x^2/\tau^2<C$ for some $C>0$, so that the second term in the right-hand side of \eqref{myeq4.5} is also bounded uniformly bounded with respect to $S$ and $\tau$. By Lemma \ref{Graf} (3) and Assumption \ref{as4} $\tilde r^b_{jk}(x^b)\to0$ as $S\to+\infty$ on $\mathrm{supp}\, \nabla^bG\left(\frac{x}{S}\right)$, so that $\tilde{\mathcal R}_{jk}\hat\chi(-M)\to 0$ as $S\to +\infty$. Thus for $0<\tilde \epsilon<\gamma_0'-\gamma_0''$ if we choose sufficiently large $S$, we have $(\hat \chi(-M)\mathbf{1})\tilde{\mathcal R}(\hat \chi(-M)\mathbf{1})\geq -\tilde \epsilon\mathbf{1}$. 

Since $\gamma_0''$ and $\gamma_0'$ can be arbitrarily close to $2d(\lambda)$, for any  $\epsilon$ if $S$ is sufficiently large, as in the proof or Theorem \ref{main} (ii) we have
$$\left(\chi\left(\frac{S^2}{2t^2}G\left(\frac{x}{S}\right)-d(\lambda)<-\epsilon\right)\mathbf{1}\right)e^{-itP}f(P)\langle \mathcal A\rangle^{-s(1+\kappa)}=O(t^{-s}),$$
for any $\epsilon,\kappa,s>0$. Using the second inequality in Lemma \ref{Graf} (1) we obtain the first estimate in Theorem \ref{main2} (iii). The proof of the second estimate is completely the same as that of Theorem \ref{main} (ii).
\end{proof}

\begin{ack}
This work was supported by JSPS KAKENHI Grant Number JP16J05967.
\end{ack}

\appendix

\section{}
In this appendix we consider the essential spectrum of $P$.
Under the Assumptions \ref{as3} we can determine the essential spectra of $P_j$.
\begin{pro}
Suppose the Assumption \ref{as3}. Then 
$\sigma_{\mathrm{ess}}(P_j)=[\Sigma_j,+\infty)$ where $\Sigma_j:=\min\limits_{\omega\in S^{n-1}}\tilde V_j(\omega)$.
\end{pro}
\begin{proof}

Assume $\lambda\geq \Sigma_j$. Noting for some $\omega_0$ we have $\Sigma_j=\tilde V_j(\omega_0)$, we take $u^k=k^{-n/2}e^{i(\lambda-\Sigma_j)^{1/2}x}\phi\left(\frac{x-k^2\omega_0}{k}\right)$, where $\phi\in C_0^{\infty}(\mathbb R^n)$ satisfy $\lVert\phi\rVert=1$. Then it is easy to see that $u^k$ is a Weyl sequence for $\lambda$ of $P_j$ (see, e.g., \cite{HiS}).

On the contrary, if $\lambda<\Sigma$, then there exist $\tilde W\in C_0^{\infty}(\mathbb R^n)$ and $c>0$ such that $P_j+\tilde W-\lambda>c$ and $\lambda$ is not a spectrum of $P_j+\tilde W$. By Weyl's theorem $\sigma_{\mathrm{ess}}(P_j)=\sigma_{\mathrm{ess}}(P_j+\tilde W)$, so that $\lambda\notin\sigma_{\mathrm{ess}}(P_j)$.
\end{proof}

We can determine the essential spectrum of $P$ from those of $P_j$.
\begin{pro}
Suppose the Assumptions \ref{as1} and \ref{as2} or \ref{as3}. Then $\sigma_{\mathrm{ess}}(P)=\bigcup_{j=1}^m\sigma_{\mathrm{ess}}(P_j)$.
\end{pro}

\begin{proof}
Set $P_0:=\mathrm{diag}(P_1,\dots,P_m)$. First, we shall prove $\sigma_{\mathrm{ess}}(P_0)=\bigcup_{j=1}^m\sigma_{\mathrm{ess}}(P_j)$. If $\lambda\in \sigma_{\mathrm{ess}}(P_j)$ for some $j$, we can find a Weyl sequence $u^k$ for $P_j$ and $\lambda$. Then the sequence $\tilde u^k\in \bigoplus_{l=1}^{m}L^2(\mathbb R^n)$ whose $j$-th element is $u^k$ and the others are $0$ is a Weyl sequence for $P_0$ and $\lambda$.

Conversely, if $\tilde u^k=(u_1^k,\dots,u_m^k)$ is a Weyl sequence for $P_0$ and $\lambda$, we have $\lVert u_j^k\rVert\geq 1/\sqrt{m}$  for some $j$ and for infinitely many $k$. Then it is easy to see that choosing a subsequence $u_j^{k_l}$ of $u_j^k$, $u_j^{k_l}/\lVert u_j^{k_l}\rVert$ is a Weyl sequence for $P_j$ and $\lambda$. 

We can write $P=P_0+R(x,D_x)$ and it is easy to see that $R(x,D_x)$ is $P_0$-compact. Thus by the Weyl's theorem $\sigma_{\mathrm{ess}}(P)=\sigma_{\mathrm{ess}}(P_0)$ which completes the proof.
\end{proof}

\bibliography{your_bib_file}

\end{document}